\documentclass[journal]{IEEEtran}%

\usepackage{setspace}
\usepackage{algorithm}
\usepackage{algorithmic}
\usepackage{amsmath,epsfig}
\usepackage{cite}
\usepackage{graphicx}
\usepackage{array}
\usepackage[latin1]{inputenc}
\usepackage{fancyhdr}
\usepackage{epsf,epsfig}
\usepackage[latin1]{inputenc}
\usepackage{amsbsy,amssymb}
\usepackage{rotating}
\usepackage{longtable}
\usepackage{amssymb}
\usepackage{bm}
\usepackage{amsmath}
\usepackage{amsfonts}
\usepackage{subfigure}
\usepackage{enumerate}
\usepackage{amssymb}
\usepackage{bm}
\usepackage{amsmath}
\usepackage{amsfonts}
\usepackage{psfrag}
\usepackage{algorithm}
\usepackage{algorithmic}
\usepackage{amsmath,epsfig}
\usepackage{cite}
\usepackage{graphicx}
\usepackage{array}
\usepackage[latin1]{inputenc}
\usepackage{epsf,epsfig}
\usepackage{amsbsy,amssymb}
\usepackage{rotating}
\usepackage{longtable}
\usepackage{amssymb}
\usepackage{bm}
\usepackage{psfrag}
\usepackage{amsmath}
\usepackage{amsfonts}
\usepackage{subfigure}
\usepackage{afterpage}
\usepackage{rotating}
\usepackage{multirow}
\usepackage{enumerate}
\usepackage{dsfont}
\usepackage{soul}%
\newtheorem{theorem}{Theorem}

\newtheorem{definition}{Definition}
\newtheorem{lemma}{Lemma}

\newtheorem{remark}{Remark}[theorem]

\def\symbolfootnote[#1]#2{\begingroup
\def\thefootnote{\fnsymbol{footnote}}
\footnote[#1]{#2}\endgroup}
\begin{document}

\title{Smart Meter Privacy for Multiple Users \\in the Presence of an \\Alternative Energy Source}

\author{\small \authorblockN{Jesus Gomez-Vilardebo$^*$ and Deniz G\"{u}nd\"{u}z$^\dagger$} \\
\authorblockA{$^*$ Centre Tecnol\`ogic de Telecomunicacions
de Catalunya (CTTC), Castelldefels, Spain \\
$^\dagger$ Imperial College London,  London, UK \\
jesus.gomez@cttc.es, d.gunduz@imperial.ac.uk}
}
\maketitle

\begin{abstract}
Smart meters (SMs) measure and report users' energy consumption to the utility provider (UP) in almost real-time, providing a much more detailed depiction of the consumer's energy consumption compared to their analog counterparts. This increased rate of information flow to the UP, together with its many potential benefits, raise important concerns regarding user privacy. This work investigates, from an information theoretic perspective, the privacy that can be achieved in a multi-user SM system in the presence of an alternative energy source (AES). To measure privacy, we use the mutual information rate between the users' real energy consumption profile and the SM readings that are available to the UP. The objective is to characterize the \textit{privacy-power function}, defined as the minimal information leakage rate that can be obtained with an average power limited AES. We characterize the privacy-power function in a single-letter form when the users' energy demands are assumed to be independent and identically distributed over time. Moreover, for binary and exponentially distributed energy demands, we provide an explicit characterization of the privacy-power function. For any discrete energy demands, we demonstrate that the privacy-power function can always be efficiently evaluated numerically. Finally, for continuous energy demands, we derive an explicit lower-bound on the privacy-power function, which is tight for exponentially distributed loads.

\end{abstract}

\section{Introduction}

\symbolfootnote[0]{This work was partially supported by the Catalan Government
under SGR2009SGR1046 and by the Spanish Government under projects
TEC2010-17816 and TSI-020400-2011-18.}

With the adoption of smart meters (SMs) in energy distribution networks the
utility providers (UPs) are able to monitor the grid more closely, and predict
the changes in the demand more accurately. This, in turn, allows the UPs to
increase the efficiency and the reliability of the grid by dynamically
adjusting the energy generation and distribution, as well as the prices,
thereby, also influencing the user demand. SMs also benefit the users
by allowing them to monitor their own energy consumption profile in almost
real time. Consumers can use this information to cut unnecessary consumption,
or to reduce the cost by dynamically shifting consumption based on the prices
dynamically set by the UPs.

SM deployment is spreading rapidly worldwide \cite{Wunderlich:AMCIS:12}. In Europe,
the adoption of SMs has been mandated by a directive of the European
Parliament \cite{EUD:09}, which requires 80\% SM adoption in all European
households by 2020 and 100\% by 2022. However, the massive deployment of SMs
at homes have also raised serious concerns regarding user privacy
\cite{McDaniel:SP:09}. High resolution SM readings can allow anyone who has
access to this data to infer valuable private information regarding user
behaviour, including the type of electrical equipments used, the time,
frequency and duration of usage \cite{Predunzi:PESWM:02}, and even the TV
channel that is being watched, as reported in \cite{Greveler:CPDP:12}. The privacy of smart meter data is more critical for businesses, such as data centers, factories, etc., whose energy consumption behaviour can reveal important information about their business to competitors. As pointed out in \cite{Molina-Markham10}, depending on the monitoring granularity different consumption patterns can be identified. With a granularity of hours or minutes, one can identify the user's presence, with a granularity of minutes or seconds one can infer the activities of appliances such as TV or refrigerator, and with a granularity of seconds one could detect bursts of power and identify the activity of appliances such as microwaves, coffee machines or toasters.

Several methods have been proposed in the literature to provide privacy to SM users while
keeping the benefits of SMs for control and
monitoring of the grid. In \cite{Efthymiou:SmartGridComm:10} user
anonymization is proposed by the participation of a trusted third party. Bohli
et al. \cite{Bohli:ICC:10} propose sending the aggregated energy consumption
of a group of users and in \cite{Wang:TSG:12} users protect their privacy by
adding random noise to their SM readings before being forwarded to the UP.
Similarly, \cite{Sankar:SmartGrid:TSG:13} proposes quantization of SM readings.

In all of the above work, privacy is obtained by distorting/transforming
the SM readings before being forwarded to the UP. However,
energy is provided to the user by the UP, and in principle, the UP can easily track user's
energy consumption by installing its own smart measurement devices at points
where the user connects to the grid. It seems that no level of privacy can be
achieved under such a strong assumption; however, users can conceal the
patterns corresponding to individual devices and usage patterns by
manipulating their energy consumption. This can be achieved either by
filtering the energy consumption over time by means of a storage device such as an electric car battery \cite{Kalogridis:SmartGridComm:10}, \cite{Varodayan:ICASP:11},
\cite{Yang:CCCS:12} and \cite{Tan:JSAC:13}, or by considering the availability
of an alternative energy source (AES)  \cite{Tan:JSAC:13},
\cite{Gunduz:ICC:13}. An AES can model a connection to a second energy grid, such as a microgrid, or a renewable energy source, such as a solar panel.

In our model, we assume that the users can satisfy part of their energy demand
from the AES. While the UP can track the energy it provides to the users
perfectly, it does not have access to the instantaneous values of the amount
of energy the user receives from the AES. Hence, 
a certain level of privacy will be achieved
depending on the amount of power available from the AES. For instance, if
the power that the AES can provide is sufficient enough to satisfy, at any time, all the
energy demand of the appliances, the privacy problem can be resolved in a
straightforward manner, as no power is requested from the power grid. However, in general, the AES will be limited in terms
of the average power it can support, and as we show in this paper, how the
user utilizes the energy provided by the AES is critical from the privacy
perspective. We measure the privacy through the mutual information rate
between the user's real energy consumption and the energy provided by the UP
(the SM readings). Mutual information has previously been proposed as a
measure of privacy in several works \cite{Agarwal:SPDS:01,RB:TKDE:10,SankarIEEETIFS13}, and in
particular, for SM systems in~\cite{Sankar:SmartGrid:TSG:13},
~\cite{Varodayan:ICASP:11} and \cite{Tan:JSAC:13}.

In our previous work \cite{Gunduz:ICC:13}, \cite{Gomez:ISIT:13} we have
characterized the minimum information leakage rate in the case of a single
user with an average and peak power constrained AES. We have shown that there
is a very close connection with this problem and the rate-distortion problem
in lossy source compression \cite{Cover:book} albeit with significant
differences. Here we generalize our results to multiple users. In this
scenario (see Fig. \ref{fig:scenario}), multiple users, each with its own
independent energy demand, share a single AES. The reason for users to share an AES can be economical. AESs, such as solar panels, and efficient energy storage units are expensive facilities, and may be shared by multiple parties to reduce cost. There could be also energy efficiency reasons: consider a scenario in which multiple smart meters belong to the same user; for example, different buildings of the same company. In such a case, the most energy efficient solution requires the centralized management of the AES for all the components of the system.

We assume that there is one
separate SM for each user, and the privacy is measured by the total
information leaked to the UP about the users' energy consumption. A single
energy management unit (EMU) receives users' instantaneous energy demands and
decides how much energy to provide to each user from the AES, while satisfying
the average power constraint. We first introduce the \textit{privacy-power
function} which characterizes the minimal information leakage rate to the UP
for a given AES average power constraint. We then provide a single-letter
information theoretic characterization of the privacy-power function for the
multi-user scenario when the input loads are independent and identically
distributed (i.i.d.) random variables. While the EMU can employ energy
management policies with memory, our result shows that a memoryless
energy management policy that randomly requests energy from the AES is
optimal, significantly simplifying the implementation.

We consider both discrete and continuous input loads. For discrete input load distributions, we first show that the optimal output alphabet can be limited to the input alphabet without loss of optimality, which allows us to write the privacy-power function as the solution of a convex optimization problem with linear constraints. As a result, the privacy-power function with discrete input loads can be evaluated numerically in polynomial time. We also provide a closed-form expression for the privacy-power function when the input loads are independent and binary distributed. Using numerical optimization, we compare the optimal privacy-power function with two heuristic power allocation schemes. 
We consider a time-division heuristic scheme which, at each time instant, obtains the requested energy either from the grid or from the AES, but not from both simultaneously. We also consider an output load limiting heuristic scheme which limits the output load to a fixed maximum value in order to cover up any variation in the energy demand beyond this value. We numerically show that our optimal scheme provides significant privacy gains compared to these heuristic energy management policies.

While the numerical evaluation of the privacy-power function for general continuous input load distributions is elusive, we derive the Shannon lower bound (SLB) on the privacy-power function, and show that this lower bound is tight when users have independent exponentially distributed input loads. For the latter case, we also show that the optimal allocation of the energy generated by the AES among the users can be obtained by the \textit{reverse waterfilling} algorithm \cite{Cover:book}. The users with low average input load satisfy all their demand from the AES, while the users with higher average load receive the same amount of energy from the grid.

The rest of the paper is organized as follows. In Section \ref{sec:Model}, we
introduce the system model, and provide a single-letter information theoretic
characterization of the privacy-power function when users have  i.i.d.
energy demands over time. Then we show that the privacy-power function for
independent users can be solved by simply minimizing the sum of the individual
privacy-power functions with a sum average power constraint. The derivation of
the privacy-power function for discrete input loads and its particularization
to binary input loads is addressed in Section \ref{sec:discrete}. Then in
Section \ref{sec:continous} the privacy-power function for continuous input
loads is studied and particularized to the exponential distribution. Numerical
results are provided in Section \ref{sec:num}. Finally, conclusions are drawn
in Section \ref{sec:con}.

\section{System Model}

\label{sec:Model}

\begin{figure*}[ptb]
\centering
\includegraphics[width=6in]{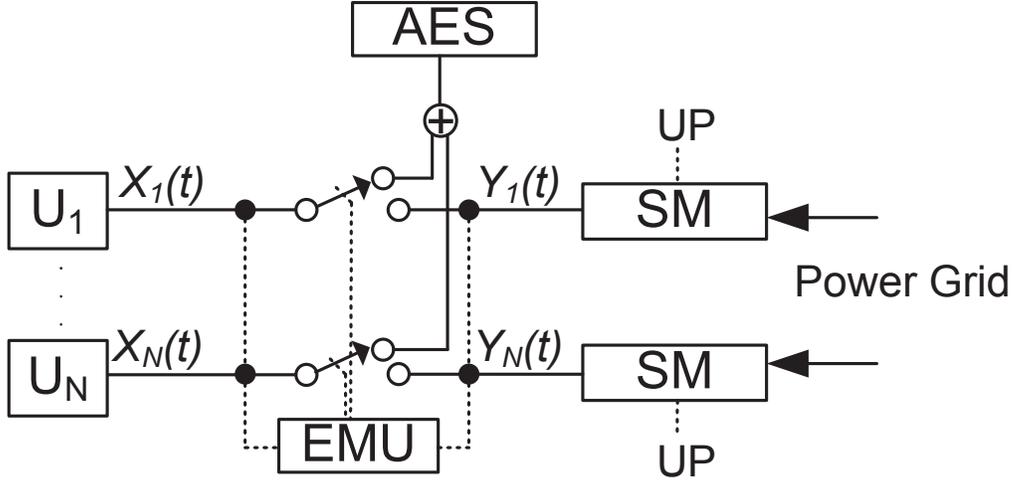}\caption{Smart meter model
studied in this paper. The EMU receives the energy demand from multiple users,
$U_{1}, \ldots, U_{N}$, and decides how much of the energy demand of each user
should be provided from the AES. The remainder of the energy demands are
satisfied from the grid, which are measured and reported by the SMs to the UP.
The privacy is measured through the information leakage rate, which measures
how much information the UP receives about the input load $[\mathbf{X}(1),
\ldots, \mathbf{X}(n)]$ by observing the SM readings $[\mathbf{Y}(1), \ldots,
\mathbf{Y}(n)]$.}%
\label{fig:scenario}%
\end{figure*}

We consider the discrete time SM model depicted in Fig. \ref{fig:scenario}. We
have $N$ users connected to the energy grid. The energy requested by user $i$
at time instant $t$ is denoted by $X_{i}(t) \in\mathcal{X}_{i}$, where
$\mathcal{X}_{i}$ is the support set of the energy demand of user $i$. We
consider the availability of an AES in the system. The AES can provide energy
to the users at a maximum average power of $P$. The AES reduces the energy
requested from the grid; but the primary use of the AES here is to create
privacy against the UP and other third parties.

The energy flow in the system is managed by the EMU. The EMU receives, at time
$t$, the energy demands of all the users, i.e., the vector $\mathbf{X}%
(t)=[X_{1}(t),...,X_{N}(t)]$. Part of the energy demand of the users can be
supported by the AES, while the remainder is provided directly from the energy
grid. We denote by $Y_{i}(t)\in\mathcal{Y}_{i}$, the amount of energy user $i$
gets from the grid at time $t$, or equivalently, the reading of SM $i$ at time
$t$. We define $\mathbf{Y}(t)=[Y_{1}(1),...,Y_{N}(t)]$ as the aggregated SM
readings available to the UP at time $t$. The energy demand of each user has
to be satisfied fully at any time, that is, we do not allow outages or
delaying/shifting the user demand. Moreover, we do not allow increasing
privacy at the expense of wasting energy, i.e., we have $0 \leq Y_{i}(t) \leq
X_{i}(t)$ for all $t$.

At the EMU, we consider energy management policies which, at each time instant
$t$, decide on the amount of power that will be provided from the AES to each
of the users based on the input loads up to time $t$, $\mathbf{X}%
^{t}=[\mathbf{X}(1),...,\mathbf{X}(t)]$, and the output loads up to the
previous time instant, $\mathbf{Y}^{t-1}=[\mathbf{Y}(1),...,\mathbf{Y}(t-1)]$.
We allow stochastic energy management policies, that is, the output load at
time $t$, $\mathbf{Y}(t)$, can be a random function of $\mathbf{X}^{t}$ and
$\mathbf{Y}^{t-1}$. We assume that, while the UP knows $P$, the average power
generated by the AES, it does not have access to the instantaneous values of
the energy users receive from the AES.

\begin{definition}
Denote the vector of input and output load alphabets for all the users as
$\mathcal{X}^{N~}=[\mathcal{X}_{1},...,\mathcal{X}_{N}]$ and $\mathcal{Y}%
^{N~}=[\mathcal{Y}_{1},...,\mathcal{Y}_{N}]$, respectively. A length-$n$
\textit{energy management policy} is composed of, possibly stochastic, power
allocation functions
\begin{equation}
f_{t}:\mathcal{X}^{N~\times~t}\times\mathcal{Y}^{N~\times~(t-1)}%
\rightarrow\mathcal{Y}^{N},
\end{equation}
for $t=1,...,n,$ such that
\begin{equation}
\mathbf{Y}(t)=f_{t}(\mathbf{X}(1),\ldots,\mathbf{X}(t),\mathbf{Y}%
(1),\ldots,\mathbf{Y}(t-1)),
\end{equation}
with $X_{i}(t)\geq Y_{i}(t)\geq0$ for all $1\leq i\leq N$ and $1\leq t\leq n$.
\end{definition}

We measure the privacy achieved by an $n-$length energy management policy with
the \textit{information leakage rate}. Assuming that the statistical behavior
of the energy demand is known by the UP, its initial uncertainty about the
real energy consumption can be measured by the entropy rate $\frac{1}%
{n}H(\mathbf{X}^{n})$. This uncertainty is reduced to $\frac{1}{n}%
H(\mathbf{X}^{n}|\mathbf{Y}^{n})$ once the UP observes the output load. Hence,
the information leaked to the UP can be measured by the reduction in the
uncertainty, or equivalently, by the mutual information rate between the input
and the output loads,
\begin{equation}
I_{n} \triangleq\frac{1}{n}I\left(  \mathbf{X}^{n};\mathbf{Y}^{n}\right)  .
\end{equation}
Notice that if we could provide all the energy required by the users from the
AES, we could achieve perfect privacy, i.e., we would have $I_{n}=0$ for all
$n$, by letting $Y_{i}(t)=0$ for all $i$ and $t$. However, in general the AES
will be limited in terms of the average power it can provide.


We are thus interested in characterizing the \emph{achievable} level of
privacy as a function of the average power $P$ that is provided by the AES,
given by
\begin{equation}
P_{n}=\mathds{E}\left[  \overset{N}{\sum_{i=1}}\frac{1}{n}\overset{n}%
{\sum_{t=1}} (X_{i}(t)-Y_{i}(t)) \right]  ,
\end{equation}
where the expectation is take over the joint probability distribution of the
input and output loads.

\begin{definition}
An information leakage rate - average power pair $(I,P)$ is said to be
\textit{achievable} if there exists a sequence of energy management policies
of duration $n$ with $\lim_{n\rightarrow\infty}I_{n}\leq I$, and
$\lim_{n\rightarrow\infty}P_{n}\leq P$.
\end{definition}

\begin{definition}
The \textit{privacy-power function}, $\mathcal{I}(P)$, is the infimum of the
information leakage rates $I$ such that $(I,P)$ is achievable.
\end{definition}

The privacy-power function characterizes the level of privacy that can be
achieved by an average power limited AES. The goal of the EMU is to achieve
the minimum information leakage rate by optimally allocating the limited
energy from the AES over the users and time.

This model of an AES is appropriate for energy sources with their own large
energy storage unit, which can provide energy reliably at a certain rate for a
sufficiently long duration of time. A peak power constraint on the AES, in
addition to the average power constraint, is also considered in
\cite{Gomez:ISIT:13}. On the other hand, in \cite{Tan:JSAC:13} we have
explicitly considered the energy generation process at the AES, in which case
the EMU is limited not only by the average power it can pull from the AES, but
also the generated energy plus the energy available in the battery at each
time instant. Such instantaneous constraints that vary over time depending on
the energy management policy and the energy arrival process at the AES, render
the analysis significantly harder as they prevent us from invoking
information theoretic arguments that will be instrumental in obtaining the
single-letter results in this work.

Our goal here is to give a mathematically tractable expression for the
privacy-power function, and identify the optimal energy management policy that
achieves it. In the rest of the paper, we consider i.i.d. input
loads for simplicity, as this will allow us to obtain a single-letter expression for the privacy-power function. Note that in most real-life applications there is a significant correlation among energy demands over time. The i.i.d. assumption allows us to characterize the optimal privacy-preserving solutions, which will be instrumental in identifying solutions for more realistic energy consumption models. Moreover, the i.i.d input load model might be valid in scenarios where the energy consumption either does not have memory at any time scale, or can be modelled as i.i.d. over the time scale of interest. This could be the case, for example, when there is a huge number of applications in use at any time, e.g., in a data center, where the input load can be modelled as i.i.d. over time for different traffic/ load states.

In the next theorem, we show that if the input load vectors
$\mathbf{X}(t)$ are i.i.d. over time with $f_{\mathbf{X}}(\mathbf{x})$, we can
characterize the function $\mathcal{I}(P)$ in a single-letter format. Note
that the instantaneous energy demands of the users can be correlated with each other.

\begin{theorem}
\label{t:EhTp_function} The privacy-power function $\mathcal{I}(P)$ for an
i.i.d. input load vector $\mathbf{X}=[X_{1},\ldots,X_{N}]$ with distribution
$f_{\mathbf{X}}(\mathbf{x})$ is given by
\begin{equation}
\mathcal{I}(P)=\inf_{\substack{f_{\mathbf{Y|X}}(\mathbf{y|x}%
):\mathds{E}\left[  \sum_{i=1}^{N}(X_{i}-Y_{i})\right]  \leq P,\\0\leq
Y_{i}\leq X_{i},\text{ }i=1,..N}}I(\mathbf{X};\mathbf{Y}), \label{ProblemF}%
\end{equation}
where $\mathbf{Y}=[Y_{1},\ldots,Y_{N}]$ is the corresponding vector of SM readings.
\end{theorem}

Some basic properties of the privacy-power function $\mathcal{I}(P)$ are
characterized in the following lemma. The proof follows from standard techniques based on time-sharing arguments \cite{Cover:book}.

\begin{lemma}
\label{l:convexity} The privacy-power function $\mathcal{I}(P)$, given above,
is a non-increasing convex function of $P$.
\end{lemma}

Next we prove Theorem \ref{t:EhTp_function}.

\begin{proof}
We first prove the achievability. Given a conditional probability distribution
$f_{\mathbf{Y|X}}(\mathbf{y|x})$ that satisfies (\ref{ProblemF}), we generate
each $\mathbf{Y}(t)$ independently using $f_{\mathbf{Y|X}}(\mathbf{y}%
(t)\mathbf{|x}(t))$. The mutual information leakage rate is then given by
$I(\mathbf{X};\mathbf{Y})$ whereas the average power constraint in
(\ref{ProblemF}) is trivially satisfied.

For the converse, assume that there is an $n-$length energy management policy
that satisfies the\ instantaneous and average constraints in (\ref{ProblemF}).
Let $H(\mathbf{X})$ denote the entropy of the random variable $\mathbf{X}$.
The information leakage rate of the resulting output load vector will satisfy
the following chain of inequalities:
\begin{subequations}
\begin{align}
\frac{1}{n}I(\mathbf{X}^{n};\mathbf{Y}^{n})  &  =\frac{1}{n}\left[
H(\mathbf{X}^{n})-H(\mathbf{X}^{n}|\mathbf{Y}^{n})\right]  , \label{conv_eq_1}%
\\
&  =\frac{1}{n}\sum_{t=1}^{n}\left[  H(\mathbf{X}(t))-H(\mathbf{X}%
(t)|\mathbf{X}^{t-1}\mathbf{Y}^{n})\right]  ,\label{conv_eq_2}\\
&  \geq\frac{1}{n}\sum_{t=1}^{n}\left[  H(\mathbf{X}(t))-H(\mathbf{X}%
(t)|\mathbf{Y}(t))\right]  ,\label{conv_eq_3}\\
&  =\frac{1}{n}\sum_{t=1}^{n}I(\mathbf{X}(t);\mathbf{Y}(t)), \label{conv_eq_4}%
\\
&  \geq\frac{1}{n}\sum_{t=1}^{n}\mathcal{I}\left(  \mathds{E}\left[
\sum_{i=1}^{N} X_{i}(t) - Y_{i}(t) \right]  \right)  ,\label{conv_eq_5}\\
&  \geq\mathcal{I}\left(  \frac{1}{n}\sum_{t=1}^{n}\mathds{E}\left[
\sum_{i=1}^{N} X_{i}(t) - Y_{i}(t) \right]  \right)  ,\label{conv_eq_6}\\
&  \geq\mathcal{I}(P), \label{conv_eq_7}%
\end{align}
where (\ref{conv_eq_2}) follows from the assumption that the input loads are
i.i.d. over time, (\ref{conv_eq_3}) follows as conditioning reduces entropy;
(\ref{conv_eq_5}) follows from the definition of the privacy-power function
$\mathcal{I}(\cdot)$; (\ref{conv_eq_6}) follows from the convexity of function
$\mathcal{I}(\cdot)$ stated in Lemma \ref{l:convexity} and Jensen's
inequality; and finally (\ref{conv_eq_7}) follows since the energy management
policy has to satisfy the average power constraint and $\mathcal{I}(\cdot)$ is
a non-increasing function of its argument.
\end{subequations}
\end{proof}

\begin{remark}
The achievability part of the proof reveals that the optimal energy management
policy is memoryless; that is, it can be achieved by simply looking at the
instantaneous input load, and generating the output load randomly using the
optimal conditional probability, which simplifies the operation of the EMU
significantly. This results in a stochastic energy management policy rather
than a deterministic one.
\end{remark}

We note here that the same performance in Theorem \ref{t:EhTp_function} can
also be achieved by a deterministic block-based energy management policy if
the user knew all the future energy demands over a block of $n$ time instants.

We also note the similarity between the privacy-power function in
(\ref{ProblemF}) and the classical rate-distortion function \cite{Cover:book}.
The characterization of the privacy-power function for a multi-user SM system
is equivalent to the rate-distortion function for a vector source with a
difference distortion measure%
\begin{equation}
d(\mathbf{x}, \mathbf{y})=\left\{
\begin{array}
[c]{ll}%
\sum_{i=1}^{N}x_{i}-y_{i}, & \text{if }y_{i}\leq x_{i},\text{ }\forall i\\
\infty, & \text{otherwise.}%
\end{array}
\right.
\end{equation}
However, despite the similarity between the expressions of the rate-distortion
and the privacy-power functions, their operational definitions are quite
different. In the case of lossy source compression, there is an encoder and a
decoder and the rate-distortion function characterizes the minimum number of
bits per sample that the encoder should send to the decoder, such that the
decoder can reconstruct the source sequence within the specified average
distortion level. In lossy source compression, the encoder observes the whole
block of $n$ source samples, and maps them to an index from the compression
codebook, which is agreed upon in advance.

There are major differences between the two problems. In the SM privacy
problem, there is neither an agreed codebook nor a digital interface. Here
$\mathbf{Y}^{n}$ is the direct output of the \textquotedblleft
encoder\textquotedblright, rather than the reconstruction of the decoder based
on the transmitted index. The EMU does not operate over blocks of input load
realizations; instead, the output load is decided instantaneously based on the
previous input and output loads. Similarly, in the SM privacy problem, there
is no encoder or decoder either, although the EMU can be considered as an
encoder and $\mathbf{Y}^{n}$ as the reconstruction of the input load
$\mathbf{X}^{n}$. However, the ``distortion'' constraint between the input and
output loads in the SM privacy problem stems from the constraint on the
available power that the AES can generate, rather than the limited rate of
encoding as in the rate - distortion problem.

Having clarified the distinctions between the privacy-power and
rate-distortion functions, we also remark the differences between our
formulation of the SM privacy problem and the \textit{privacy-utility
framework} studied in \cite{Sankar:SmartGrid:TSG:13}. In our privacy
model the SM readings are not tempered, and thus, the SMs report the exact amount of energy
received from the grid. On the other hand, in \cite{Sankar:SmartGrid:TSG:13}, the SM readings are
considered as the samples of an information source, which are compressed
before being forwarded to the UP in order to hide their real values; and
hence, privacy is achieved at the expense of distorting the SM measurements. The
distortion constraint in \cite{Sankar:SmartGrid:TSG:13} is explicit and
measures the utility of the compressed SM samples.

If the users' input loads are independent from each other, but not necessarily
identically distributed, the multi-user privacy-power function in
(\ref{ProblemF}) simplifies further. The following chain of inequalities lower
bound the privacy-power function under this assumption:
\begin{subequations}
\begin{align}
I(\mathbf{X};\mathbf{Y})  &  =\sum_{i=1}^{N}H(X_{i})-H\left(  X_{i}%
|X^{i-1},Y^{N}\right)  ,\\
&  \geq\underset{i=1}{\overset{N}{\sum}}H\left(  X_{i}\right)  -\underset
{i=1}{\overset{N}{\sum}}H\left(  X_{i}|Y_{i}\right)  ,\label{indep}\\
&  =\underset{i=1}{\overset{N}{\sum}}I\left(  X_{i};Y_{i}\right)  ,\\
&  \geq\sum_{i=1}^{N}\mathcal{I}_{X_{i}}\left(  P_{i}\right)  ,
\label{single_apliances}%
\end{align}
where we have defined $P_{i}=\mathds{E}[X_i-Y_i]$, and $\mathcal{I}_{X_{i}%
}(\cdot)$ denotes the privacy power function for a system with an input load
distribution $f_{X_{i}}(x_{i})$. We can achieve equality in (\ref{indep}) with
independent EMU policies for individual users, $f_{\mathbf{Y}|\mathbf{X}}
(\mathbf{y}|\mathbf{x}) = \prod_{i}^{N}f_{Y_{i}|X_{i}}(y_{i} |x_{i})$. Consequently, we can achieve equality in (\ref{single_apliances}) by using the single user
optimal energy management policy for each of the input loads separately, while
satisfying the total average power constraint, $\sum_{i=1}^{N}P_{i}\leq P$.


Following the above arguments, the problem of characterizing the optimal
privacy-power function for a multi-user SM system is reduced to the following
optimization problem
\end{subequations}
\begin{equation}
\mathcal{I}(P)=\underset{\sum_{i=1}^{N}P_{i}\leq P}{\inf}\sum_{i=1}%
^{N}\mathcal{I}_{\mathsf{X}_{i}}\left(  P_{i}\right)  .
\label{independentapliances}%
\end{equation}



In the following sections, we use the information theoretic single-letter
characterization of the privacy-power function in order to obtain either
closed-form solutions or numerical algorithms that give us the optimal energy
management policies in multi-user SM systems with certain input load
distributions and an average power constraint on the AES.

\section{Discrete Input Loads}

\label{sec:discrete}

In the previous section we have characterized the privacy-power function for
i.i.d. input loads as an optimization problem in a single-letter format in
(\ref{ProblemF}). Now we will show that this problem can always be efficiently
solved for any discrete input load distribution. In addition, for the
particular case where all the users have binary input loads, we give a
closed-form expression for the privacy-power function.

For discrete input and output alphabets, the characterization of the privacy-power function $\mathcal{I}(\mathcal{P})$ in (\ref{ProblemF}) is a convex
optimization problem since the mutual information is a convex function of the
conditional probabilities, $f_{\mathbf{Y|X}}(\mathbf{y}|\mathbf{x})$, for
$\mathbf{y}\in\mathcal{Y}^{N}$, $\mathbf{x}\in\mathcal{X}^{N}$, and the
constraints are linear. Then, (\ref{ProblemF}) can be solved numerically,
e.g., by the efficient Blahut-Arimoto (BA) algorithm \cite{Cover:book}.
However, while the input load alphabet, defined by the system based on the
energy demand profiles of the users, can be discrete, the output load alphabet
is not necessarily discrete, and the output load, in general ,can take any real value. The next
theorem shows that for discrete input load alphabets, the output load alphabet
can be constrained to the input alphabet without loss of
optimally, i.e., $\mathcal{Y=X}$, and consequently, for any given discrete input alphabet the privacy-power function can always be computed efficiently. This result is only valid
for i.i.d. input loads, but does not require users' input loads to be
independent from each other.

\begin{theorem}
\label{Discrite_Output_Input_Alphabet} Without loss of optimality, for
discrete input load alphabets, the output load alphabet $\mathcal{Y}^{N}$ can
be constrained to the input load alphabet, i.e., $\mathcal{Y}^{N}%
\mathcal{=X}^{N}$.
\end{theorem}

\begin{proof}
Let the discrete input load alphabets for each user be defined as a possibly
infinite set
\[
\mathcal{X}_{i}=\{x_{i,1},...,x_{i,m_{i}}:x_{i,j}<x_{i,j+1}\},
\]
where $m_{i}=+\infty$ if the input alphabet is countably infinite.
Define $\mathcal{X}_{i}^{C}$ as the set of non-negative real numbers that are
not in the input load alphabet for each user $i$.
For any vector $\mathbf{x}=[x_{1},....,x_{N}]\in\mathcal{X}^{N}$ define the
set
\[
\Omega\left(  \mathbf{x}\right)  \triangleq(x_{1}^{-},x_{1}]\times\cdots
\times(x_{N}^{-},x_{N}]
\]
where $\times$ denotes the Cartesian product and $x_{i}^{-}=\max\left\{
x\in\left\{  0,\mathcal{X}_{i}\right\}  :x<x_{i}\right\} $. Now assume that
the optimal privacy-power function in (\ref{ProblemF}) is achieved by the
conditional probability distribution $f_{\mathbf{Y|X}}(\mathbf{y}|\mathbf{x}%
)$, which might take positive values for some $y_{i}\in\mathcal{X}_{i}^{C}$.
We define the following new conditional probability distribution:%
\[
f_{\mathbf{\hat{Y}|X}}(\mathbf{\hat{y}}|\mathbf{x})=\left\{
\begin{array}
[c]{ll}%
0, & \text{if }\exists i:\hat{y}_{i}\in\mathcal{X}_{i}^{C},\\
\int_{\Omega\left(  \mathbf{\hat{y}}\right)  }f_{\mathbf{Y|X}}(\mathbf{y}%
|\mathbf{x})d\mathbf{y}, & \text{if }\hat{y}_{i}\in\mathcal{X}_{i},~\forall i.
\end{array}
\right.
\]

The new conditional probability function does not allow any output value in
$\mathcal{X}_{i}^{C}$ for any $i$, i.e., the output alphabet is limited to the
input alphabet. Instead, any output vector $\mathbf{y}=[y_{1},\ldots,y_{N}]$,
which has a non-zero probability according to $f_{\mathbf{Y|X}}(\mathbf{y}%
|\mathbf{x})$, is assigned to a new output vector $[\hat{y}_{1},\ldots,\hat
{y}_{N}]$ such that
\begin{equation}
\label{eq:y_hat_y}\hat{y}_{i}=\min\{x\in\mathcal{X}_{i}:x\geq y_{i}\}.
\end{equation}
Notice that the energy management policy, $f_{\mathbf{\hat{Y}|X}}%
(\mathbf{\hat{y}}|\mathbf{x})$, is still feasible since the output load, at
any time instant, is still less than what is requested by the appliances, i.e.,
$\hat{y}_{i} \leq x_{i}, ~\forall i$. Moreover, with this new
conditional distribution the power load demanded from the AES can only have a
smaller average value compared to the original energy management policy, since
the output load is not reduced for any input load value. Thus, it only remains
to show that the new conditional distribution leaks at most the same amount of
information to the UP. Notice that the new output load $\mathbf{\hat{Y}}$ is a
deterministic function of $\mathbf{Y}$ define in (\ref{eq:y_hat_y}). Hence,
from the information processing inequality, we have that $\mathbf{X} -
\mathbf{Y} - \mathbf{\hat{Y}}$ form a Markov chain, and consequently,
$I(\mathbf{X},\mathbf{Y})\geq I(\mathbf{X},\mathbf{\hat{Y}})$, which completes
the proof.
\end{proof}

\subsection{Binary Input Loads}

The simplest discrete input load model we can consider is a binary input
alphabet with independent Bernoulli input load distributions for all the
users, i.e., $X_{i}\sim\mathbf{Ber}(p_{i})$, where $p_{i}=p_{X_{i}}(L_{i})$
and $\mathcal{X}_{i}=\left\{  L_{i},H_{i}\right\}  $ for $i=1,...,N$. Observe
that the average power required by the $i-$th user is given by $P_{X_{i}%
}=L_{i}+\Delta_{i}\left(  1-p_{i}\right)  $, where $\Delta_{i}=H_{i}-L_{i}$.
This power consumption model corresponds to a scenario in which the users, at
each time instant, require either a constant high power load level $H_{i}$, or
a constant low power load level $L_{i}$, i.e., the standby power consumption
level. When there is a power demand, the EMU fulfills this demand either
obtaining the energy from the UP, or from the AES according to $p_{\mathbf{Y}%
|\mathbf{X}}$.

From Theorem \ref{Discrite_Output_Input_Alphabet}, the optimal output
distribution $\mathcal{Y}_{i}$ is also binary for all $i$. Hence, the power
allocated from the AES to each user is a binary random variable over the set
$\left\{  0,\Delta_{i}\right\}  $. Note that, since we require $Y_{i} \leq
X_{i}$, we can only provide energy from the AES to user $i$ if $X_{i}%
(t)=H_{i}$ and $Y_{i}(t)=L_{i}$, and consequently, $p_{X_{i}Y_{i}}(L_{i}%
,H_{i})=0$ and $p_{X_{i}Y_{i}}(L_{i},L_{i})=p_{X_{i}}(L_{i})=p_{i}$. The
energy obtained from the AES is then directly related to $p_{X_{i}Y_{i}}%
(H_{i},L_{i})$ by $P_{i}=\Delta_{i}p_{X_{i}Y_{i}}(H_{i},L_{i}),$ and we can
express the mutual information $I\left(  X_{i};Y_{i}\right)  $ for the
bivariate binary distribution
\[
p_{X_{i}Y_{i}}=\left[
\begin{array}
[c]{cc}%
p_{i} & 0\\
\frac{P_{i}}{\Delta_{i}} & 1-p_{i}-\frac{P_{i}}{\Delta_{i}}%
\end{array}
\right]  ,
\]
as a function of $P_{i}$ as follows:%
\begin{multline*}
I_{\mathsf{B}_{i}}\left(  P_{i}\right)  =\frac{P_{i}}{\Delta_{i}}\log
_{2}\left(  \frac{P_{i}}{\Delta_{i}}\right)  -\left(  p_{i}+\frac{P_{i}%
}{\Delta_{i}}\right)  \log_{2}\left(  p_{i}+\frac{P_{i}}{\Delta_{i}}\right)
\\-\left(  1-p_{i}\right)  \log_{2}\left(  1-p_{i}\right)  .
\end{multline*}
Observe that $I_{\mathsf{B}_{i}}\left(  P_{i}\right)  $ is a monotonically
decreasing function of $P_{i}$, and $I_{\mathsf{B}_{i}}(\Delta_{i}\left(
1-p_{i}\right)  )=0$. Consequently, the privacy-power function for the binary
model for a single user is given by
\begin{equation}
\mathcal{I}_{\mathsf{B}_{i}}(P_{i})=\left(  I_{\mathsf{B}_{i}}\left(
P_{i}\right)  \right)  ^{+},
\end{equation}
where $(x)^{+}=\max(x,0)$.

By particularizing (\ref{independentapliances}) with $\mathcal{I}%
_{\mathsf{X}_{i}}(P_{i})=\mathcal{I}_{\mathsf{B}_{i}}(P_{i})$ for all $i$, and
solving the resultant problem, we find the optimal power allocation
$P_{i}^{\ast}$ as%
\begin{equation}
P_{i}^{\ast}=\left\{
\begin{array}
[c]{ll}%
\Delta_{i} p_{i} \frac{1-p_{\Delta_{i}}}{p_{\Delta_{i}}} & \text{if }%
p_{i}<p_{\Delta_{i}},\\
\Delta_{i} (1-p_{i}) & \text{otherwise},
\end{array}
\right.
\end{equation}
where $p_{\Delta_{i}}\left(  \lambda\right)  =1-e^{-\lambda\Delta_{i}}$, and
$\lambda$ is chosen such that $\sum_{i=1}^{N}P_{i}^{\ast}=P$ . Note that
$p_{\Delta_{i}}$ satisfies $0\leq p_{\Delta_{i}}\leq1$. Then, the
privacy-power function for the multiple users with independent binary input
load distributions is given by
\begin{align}
\mathcal{I}_{\mathbf{B}}(P)  &  =\sum_{i=1}^{N}\mathcal{I}_{\mathsf{B}_{i}%
}(P_{i}^{\ast}),\\
&  =\sum_{i=1}^{N}\left(  H_{\mathsf{B}}(p_{i})-\frac{p_{i}}{p_{\Delta_{i}}%
}H_{\mathsf{B}}(p_{\Delta_{i}})\right)  ^{+},
\end{align}
where $H_{\mathsf{B}}(p)$ denotes the entropy of a $\mathbf{Ber}(p)$ distribution.

Each user can achieve full privacy $\mathcal{I}_{\mathsf{B}_{i}}(P_{i}^{\ast
})=0$ by obtaining an average power of $P_{X_{i}}-L_{i} = \Delta_{i}
(1-p_{i})$ from the AES, the remaining power $L_{i}$ is obtained from the grid
without incurring any lost of privacy. However, if the average power obtained
from the AES is below $P_{X_{i}}-L_{i}$ then the energy obtained from the grid
comes at the expense of a loss in privacy. Note that $P_{i}^{\ast}$ and
$\mathcal{I}_{\mathbf{B}}(P)$ depend on the input load parameters $P_{X_{i}},$
$L_{i}$, $\Delta_{i}$, and $p_{i}$ in a non-straightforward manner. We
postpone the detailed analysis of this privacy-power function to Section
\ref{sec:num}.

\section{Continuous Input Loads}

\label{sec:continous}

For continuous input loads, the optimal output alphabet is also continuous.
Consequently, efficient algorithms, such as the BA algorithm, do not yield the
optimal solution to (\ref{ProblemF}). In this case, we provide a lower bound
on the privacy-power function by using the Shannon lower bound. We then
show that this lower bound is achievable when the users have independent
exponentially distributed input loads.

Using the SLB \cite{Cover:book}, for any input load distribution, we have
\begin{equation}
\mathcal{I}_{X_{i}}(P_{i})\geq\left(  \mathsf{h}(X_{i})-\ln\left(
P_{i}\right)  \right)  ^{+}\text{ nats}, \label{SLB}%
\end{equation}
where $\mathsf{h}(X)$ denotes the differential entropy of the continuous
random variable $X$. Observe that,
\begin{subequations}
\label{Exp_lower}%
\begin{align}
I(X_{i},Y_{i})  &  =\mathsf{h}(X_{i})-\mathsf{h}(X_{i}|Y_{i}%
),\label{Exp_lower_eq1}\\
&  =\mathsf{h}(X_{i})-\mathsf{h}(X_{i}-Y_{i}|Y_{i}),\label{Exp_lower_eq2}\\
&  \geq\mathsf{h}(X_{i})-\mathsf{h}(X_{i}-Y_{i}),\label{Exp_lower_eq3}\\
&  \geq\mathsf{h}(X_{i})-\mathsf{h}(\mathsf{Exp}(\mathds{E}\left[  X_{i}%
-Y_{i}\right]  )),\label{Exp_lower_eq5}\\
&  =\mathsf{h}(X_{i})-\ln\left(  P_{i}\right)  , \label{Exp_lower_eq7}%
\end{align}
where we have used $\mathsf{Exp}(\lambda)$ to denote an exponential random
variable with mean $\lambda$. In the above chain of inequalities,
(\ref{Exp_lower_eq3}) follows as conditioning reduces entropy, and
(\ref{Exp_lower_eq5}) follows since exponential distribution maximizes the
entropy among all nonnegative distributions with a given mean value
\cite{Cover:book}.

Next, we present the necessary and sufficient conditions for any piecewise
continuous input load distribution $f_{X}(x)$ to achieve the SLB, together
with the conditional probability distribution $f_{Y|X}(y|x)$ achieving it. We
denote by $u(x)$, the unit step function which assigns $0$ for $x<0$, and $1$
for $x\geq0$. The Dirac delta function is denoted by $\delta(x)$. We use
$f^{\prime}(x)$ to denote the first order derivative of $f(x)$ and
$f(x_{i}^{+})=\underset{x\rightarrow x_{i}^{+}}{\lim}$ $f(x)$ and $f(x_{i}%
^{-})=\underset{x\rightarrow x_{i}^{-}}{\lim}$ $f(x)$ and $x\rightarrow
x_{i}^{+}$ and $x\rightarrow x_{i}^{-}$ mean that $x\rightarrow x_{i}$ from
the left and right, respectively. Finally, we define $\Delta_{f}%
(x_{i})=f(x_{i}^{+})-f(x_{i}^{-})$.
\end{subequations}
\begin{theorem}
\label{ThCond_Bounded}Suppose that the input load distribution $f_{X}(x)$ is
continuous on $\mathcal{R}_{+}$ except for a countable number of jump
discontinuities or non-differentiable points $\mathcal{X}_{D}=\left\{
x_{1},...,x_{D}\right\}  $. Then, the SLB (\ref{SLB}) is achieved for all $P$
satisfying $g_{Y}(y)\geq0$, $\forall y\in\mathcal{R}_{+}$,
where
\begin{equation}
g_{Y}(y)=g_{Y_{C}}(y)+g_{Y_{D}}(y) \label{f_Y_T}%
\end{equation}
is a mixture of a continuous and a discrete function specified as follows:
\begin{align*}
g_{Y_{C}}(y)  &  =f_{X}(y)+\mathds{E}[V]f_{X}^{\prime}(y),\text{ }%
y\in\mathcal{R}_{+}/\mathcal{X}_{D},\\
g_{Y_{D}}(y)  &  =\mathds{E}[V]\sum_{i=0}^{D}\Delta_{X}(x_{i})\delta
(y-x_{i}),\text{ }y\in\mathcal{X}_{D}.
\end{align*}
For all $P$, at which the SLB is achieved, the output distribution is given by $f_Y(y)=g_Y(y)$ and the
optimal conditional output load distribution reads $f_{Y|X}(y|x)=f_{V}%
(x-y)\frac{f_{Y}(y)}{f_{X}(x)}$ where $f_{V}(v)=\frac{1}{\mathds{E}[V]}%
e^{-\frac{v}{\mathds{E}[V]}}u(v)$.
\end{theorem}

\begin{proof}
To show this results, we need to find the conditional distribution
$f_{Y|X}(y|x)$ that satisfies the SLB with equality \cite{Cover:book}. We
require the random variables $V=X-Y$ and $Y$ to be independent, and $V$ to be
distributed according to an exponential distribution $V\sim\mathsf{Exp}(P)$
with mean $P$. We first obtain the output distribution $f_{Y}(y)$ from its
Laplace transform $\mathcal{L}f_{Y}(s)=\mathcal{L}(f_{Y}(y))(s)$ as
\begin{align*}
\mathcal{L}f_{Y}(s) &  =\frac{\mathcal{L}f_{X}(s)}{\mathcal{L}f_{V}(s)},\\
&  =\mathcal{L}f_{X}(s)\left(  1+\mathds{E}[V]s\right)  .
\end{align*}
Then, it follows that $f_{Y}(y)$ is given by (\ref{f_Y_T}). The conditional
distribution $f_{Y|X}(y|x)$ is obtained using the fact that $f_{X|Y}%
(x|y)=f_{V}(x-y)$. Finally, it can be shown that $\int_{0}^{\infty}%
f_{Y}(y)dy=1$; and thus, the achievability is guaranteed by requiring
$f_{Y}(y)\geq0$, $\forall y\in\mathcal{R}^{+}$.
\end{proof}

\begin{remark}
If the achievability condition in Theorem \ref{ThCond_Bounded} is satisfied
for a given $P_{\max}$, it is satisfied at any $P\leq P_{\max}$. Then it
follows that, there is a unique critical average power level, $P_{0}$, such
that $\mathcal{I}_{X}(P)=\mathsf{h}(X)-\ln\left(  P\right)  $ for all $P\leq
P_{0}$ and $\mathcal{I}_{X}(P)>\mathsf{h}(X)-\ln\left(  P\right)  $ for all
$P>P_{0}$.
\end{remark}

To find a lowerbound on the privacy-power function in the case of multiple
users with continuous input load distributions, we replace $\mathcal{I}%
_{X_{i}}(P_{i})$ with $\left(  \mathsf{h}(X_{i})-\ln\left(  P_{i}\right)
\right)  ^{+}$ in (\ref{independentapliances}), and find the corresponding
optimal power allocation $P_{i}^{\ast}$ as
\begin{equation}
P_{i}^{\ast}=\left\{
\begin{array}
[c]{cl}%
\lambda, & \text{if }e^{\mathsf{h}(X_{i})}>\lambda,\\
e^{\text{ }\mathsf{h}(X_{i})}, & \text{otherwise},
\end{array}
\right.
\end{equation}
where $\lambda$ is chosen such that $\sum_{i=1}^{N}P_{i}^{\ast}=P$. Then the
privacy-power function for multiple users can be lower-bounded by
\begin{equation}
\mathcal{I}_{\mathbf{X}}(P)\geq\sum_{i=1}^{N}\left(  \mathsf{h}(X_{i}%
)-\ln\left(  \lambda\right)  \right)  ^{+}\text{ nats}.
\end{equation}

\subsection{Exponential Input Loads}

\psfrag{l1}{$\lambda_1$} \psfrag{l2}{$\lambda_2$} \psfrag{l3}{$\lambda_3$}
\psfrag{la}{$\lambda$} \begin{figure}[t]
\centering\includegraphics[width=2in]{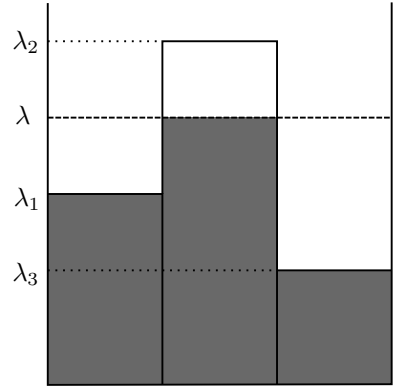}\caption{The
reverse waterfilling solution for the optimal power provided to each user from
the AES.}%
\label{fig:RWF}%
\end{figure}

For an exponential input load distribution with mean $\lambda_{i}$, i.e.,
$X_{i}\sim\mathsf{Exp}\left(  \lambda_{i}\right)  $, the SLB in (\ref{SLB}) is
achievable by using the conditional distribution \cite{Gomez:ISIT:13}
\[
f_{Y_{i}|X_{i}}\left(  y|x\right)  =\frac{\lambda_{i}}{P_{i}} e^{-\frac
{(x-y)}{P_{_{i}}}}e^{\frac{x}{\lambda_{i} }}f_{Y_{i}}(y),
\]
where $f_{Y_{i}}$ is a mixture of a continuous and a discrete distribution
specified by%
\[
f_{Y_{i}}(y)=\left(  1-\frac{P_{_{i}}}{\lambda_{i}}\right)  \frac{1}%
{\lambda_{i}}e^{-\frac{y}{\lambda_{i}}}+\frac{P_{_{i}}}{\lambda_{i}}%
\delta(y).
\]
Then the privacy-power function
for a single user with an exponential input load with mean $\lambda_{i}$ can
be explicitly characterized as follows:
\begin{equation}
\mathcal{I}_{\mathsf{E}_{i}}(P_{_{i}})= \left\{
\begin{array}
[c]{cc}%
\ln\left(  \frac{\lambda_{i}}{P_{i}} \right)  , & \text{if } P_{i} \leq
\lambda_{i},\\
0, & \text{otherwise}.
\end{array}
\right.
\end{equation}

By particularizing (\ref{independentapliances}) with $\mathcal{I}%
_{\mathsf{X}_{i}}(P_{i})=\mathcal{I}_{\mathsf{E}_{i}}(P_{_{i}})$ for all $i$,
and solving the resultant problem, we find the optimal AES power allocation
among users, $P_{i}^{\ast}$, as the well-known reverse waterfilling solution
\[
P_{i}^{\ast}=\left\{
\begin{array}
[c]{cc}%
\lambda, & \text{if }\lambda<\lambda_{i},\\
\lambda_{i}, & \text{if }\lambda\geq\lambda_{i},
\end{array}
\right.
\]
where $\lambda$ is chosen such that $\sum_{i=1}^{N}P_{_{i}}^{\ast}=P$.

The reverse waterfilling power allocation is illustrated in Fig. \ref{fig:RWF}
for three users with independent exponentially distributed energy demands with
means $\lambda_{1}, \lambda_{2}$, and $\lambda_{3}$, respectively. The optimal
reverse water level is given by $\lambda$, where the heights of the shaded
areas in the figure correspond to the average AES powers allocated to the different
users. We observe that the optimal energy management policy satisfies all the
energy demands of the users whose average input load is below $\lambda$,
directly from the AES. Hence, no information is leaked to the UP about the
energy consumption of these users; user 1 and user 3 in the figure. The rest
of the users receive exactly the same amount of power $\lambda$ from the AES,
and the remainder of their energy demand is satisfied from the grid. Finally,
the privacy-power function for multiple users with exponential input loads can
be expressed as
\begin{align}
\mathcal{I}_{\mathbf{E}}(P) =\sum_{i=1}^{N}\left( \ln\left(  \frac
{\lambda_{i}}{\lambda}\right)  \right)  ^{+}.
\end{align}

\section{Numerical Results}

\label{sec:num}

\begin{figure}[ptb]
\centering
\psfrag{Y}{\footnotesize$\mathcal{I}(P)$ bits}
\psfrag{X}{\footnotesize$P$}
\includegraphics[width=3in]{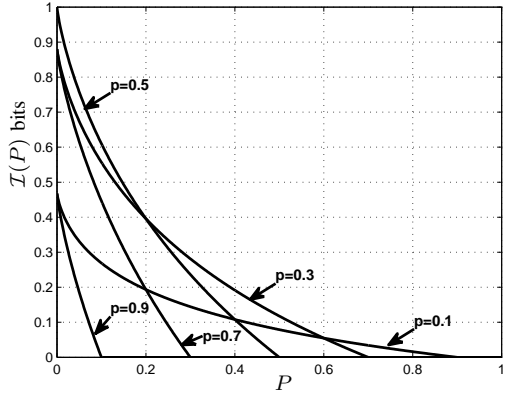}\caption{Privacy-power
function for a binary input-output system with different $p$ values.}%
\label{f:binary}%
\end{figure}

\begin{figure}[ptb]
\centering
\psfrag{I(E)}{\footnotesize$\mathcal{I}(P)$ bits}
\psfrag{E}{\footnotesize$P$}
\includegraphics[width=3in]{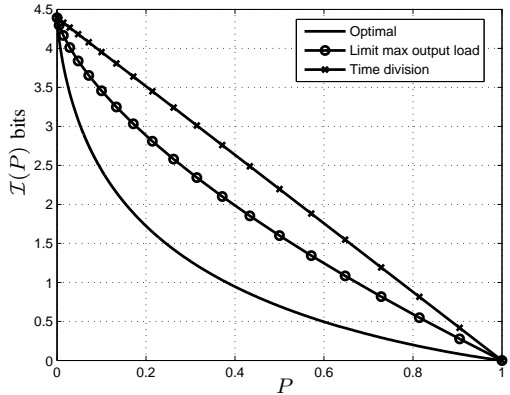}\caption{Privacy-power
function for a uniform input load, and different EMU policies.}%
\label{f:uniform}%
\end{figure}

\begin{figure}[t]
\psfrag{Y}{\footnotesize$\mathcal{I}_i(P)$ bits}
\psfrag{P}{\footnotesize$P$}
\centering\includegraphics[width=3in]{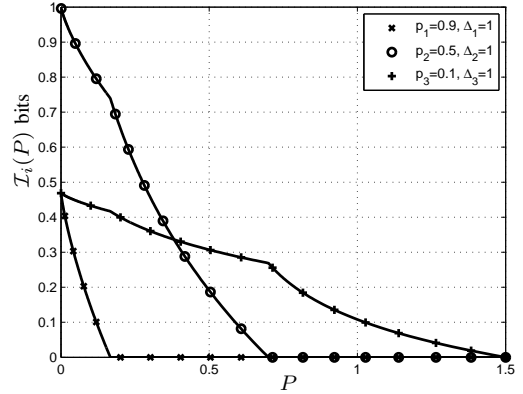}\caption{Individual privacy
achieved by three users $\mathcal{I}_{\mathsf{X}_{i}}(P)$, $i=\{1,2,3\}$, all with the same input load alphabet $\delta_i=1$ $i=\{1,2,3\}$, but with different input load distributions as a function of the average AES power $P$.}%
\label{fig:bin_I2}%
\end{figure}

In this section we numerically analyze the privacy-power function in a SM
system with various input load distributions and number of users, by
explicitly evaluating the information theoretic optimal leakage rate expressions.

\subsection{Single User Scenario}

In order to illustrate the behaviour of the privacy-power function for a
simple binary input load system, we first consider a single user with an input
load alphabet $\mathcal{X}=\mathcal{Y}=\{0,1\}$, and $p_{X}(0)=p$. We plot the
$\mathcal{I}(P)$ function for the binary input load in Fig. \ref{f:binary} for
different $p$ values. As expected, the required average power from the AES is
maximum when the user wants perfect privacy, and it is zero when no privacy is
required. We also observe that the privacy-power function is
decreasing in $P$ and convex. Another interesting observation from the figure
is the fact that the $\mathcal{I}(P)$ curves for two different input load
distributions, i.e., different $p$ values, might intersect. This means that,
to achieve the same level of privacy a lighter input load might require lower
or higher average power than a heavier input load. Also note that the two
different input load distributions, say $p=0.1$ and $p=0.9$, have the same
level of privacy when there is no AES in the system; however, the input load
with lower average energy demand, i.e., the one with $p=0.9$, achieves perfect
privacy with a much lower $P$ value.

\begin{figure}[t]
\psfrag{Y}{\footnotesize$\mathcal{I}_i(P)$ bits}
\psfrag{P}{\footnotesize$P$}
\centering\includegraphics[width=3in]{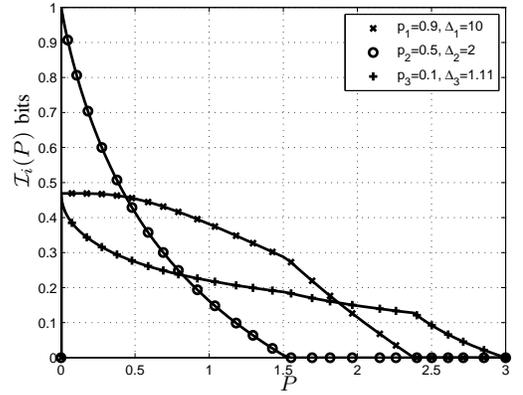}\caption{Individual privacy
achieved by three users $\mathcal{I}_{\mathsf{X}_{i}}(P)$, $i=\{1,2,3\}$, each with a different input load alphabet and input load distribution, as a function of the average AES power $P$.}%
\label{fig:bin_I}%
\end{figure}

Next, we use the discrete uniform distribution to compare the privacy
protection achieved by the information theoretical optimal policy derived here,
with different heuristic policies. In this case, the input load has a uniform distribution $U(x)$ with input load
alphabet $\mathcal{X}=\left\{  0,c,2c,...,(N-1)c\right\}  $, where $c=\frac
{2}{N-1}$ is a constant used to impose a mean value of $\mathds{E}[X]=1$.
Based on Theorem \ref{Discrite_Output_Input_Alphabet}, the output load
alphabet can be limited to $\mathcal{X}$ without loss of optimality.
We set $N=21$ and in Fig. \ref{f:uniform} we plot the privacy-power function for the
optimal strategy obtained by the BA algorithm together with
the privacy-power functions of the following two heuristic strategies:

\textbf{Time Division:} In this policy, at each time instant, the EMU gets all
the energy needed by the user, either from the AES or from the grid, but not
from both simultaneously. Then, to satisfy the average power constraint at the
AES, the EMU obtains energy from the AES with probability $\frac
{P}{\mathds{E}[X]}$. The information leaked to the UP, is thus given by%

\begin{align*}
\begin{split}
I(X;Y)   ={}&H(X)-H(X|Y=0)\frac{P}{\mathds{E}[X]}\\&-H(X|Y=x)\left(  1-\frac
{P}{\mathds{E}[X]}\right),
\end{split}\\
 ={}& \left(  1-\frac{P}{\mathds{E}[X]}\right) \log_{2}N .
\end{align*}

\textbf{Limit Maximum Output Load:} In this policy, we use the AES to
limit the maximum energy received from the grid. At each time instant, we get
all the energy from the grid $X(t)=Y(t)$ if $X(t)<kc$, whereas if $X(t)\geq
kc$ we get $Y(t)=kc$ from the grid and the remaining energy is taken from the
AES. In this case, for each $k=0,...,N-1$, the average power requested from the AES is
given by
\[
P=\left(  N-1-k\right)  (N-k)\frac{c}{2N},
\]
and the information leaked to the UP is%
\begin{align*}
I(X;Y)  & =H(X)-\Pr(Y=kc)H(X|Y=kc),\\
& =\log_{2}N-\frac{N-k}{2N}\log_{2}(N-k).
\end{align*}
In Fig. \ref{f:uniform}, we can observe that given an average power limited AES, the privacy achieved by both of these heuristics is significantly lower than that of the
optimal EMU policy.

\subsection{Multi-user Scenario}

Next we consider a multi-user scenario with $N=3$ users. We assume equal
binary load levels $H_{i}=1$ and $L_{i}=0$, but different average energy
demands with $p_{1}=0.9,$ $p_{2}=0.5,$ and $p_{3}=0.1$; thus we have
$P_{X_{1}}=0.1$, $P_{X_{2}}=0.5$, $P_{X_{3}}=0.9$. Fig. \ref{fig:bin_I2}
illustrates the privacy for each user $\mathcal{I}_{\mathsf{B}_{i}}%
(P_{i}^{\ast})$ as a function of the average power $P$ available at the AES.
Notice that, although users 1 and 3, in the absence of an AES, leak the same
amount of information to the UP, since $H_{\mathsf{B}}(0.1)=H_{\mathsf{B}%
}(0.9),$ user 1 achieves perfect privacy much more rapidly since it has a lower average
energy demand. Also note that, user 3 achieves
perfect privacy for a much higher value of $P$, even compared to user 2, which
leaks the highest amount of information when there is no AES, as it has the
highest entropy.

Remember that, as opposed to the exponential input load scenario, in the
binary case, the privacy-power function $\mathcal{I}_{\mathsf{B}_{i}}%
(P_{i}^{\ast})$ for each user does not depend solely on the average power
demand of the user, but on both of the parameters $\Delta_{i}$ and $p_{i}$. To
illustrate this dependence, we consider a scenario again with $N=3$ users, but
with equal average power demands $P_{X_{i}}=\Delta_{i}\left(  1-p_{i}\right)
$, while $L_{i}=0$ for all $i$. We choose different parameters $\Delta_{i}$
and $p_{i}$ for each user. Fig. \ref{fig:bin_I} again shows the privacy of each user as a function of the
average power $P$. Observe that the optimal power
allocation quickly reduces the information leaked by user 2, and achieves
perfect privacy for this user much before the other two, although this is the
user leaking the most amount of information in the absence of an AES. The
input power loads for users 1 and 3 have equal entropy, but with different
behaviours; user 1 demands large amounts of energy but very rarely, while user
3 demands low amounts of energy very frequently. The optimal EMU policy seen
by these users also differs significantly. While for user 1 the privacy-power
function is a concave monotonically decreasing function, for user 3 the
privacy-power function is monotonically decreasing but piecewise convex.

\begin{figure}[t]
\psfrag{I(P) bits}{\footnotesize$\mathcal{I}(P)$ bits}
\psfrag{P}{\footnotesize$P$}
\centering\includegraphics[width=3in]{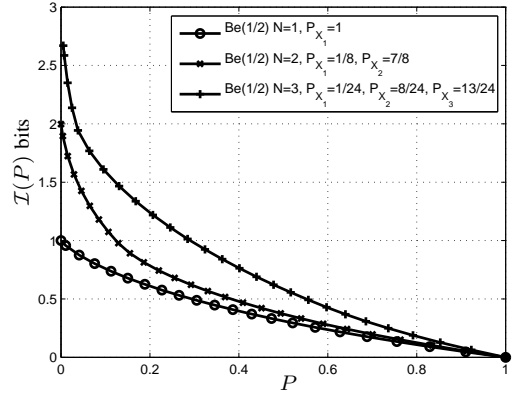}\caption{$\mathcal{I}%
(P)$ with respect to the average AES power $P$ for binary input loads with
different number of users.}%
\label{fig:bin}%
\end{figure}

Next, we study the effect of the number of users on the privacy-power
function. In Fig. \ref{fig:bin}, we depict the optimal information leakage
rate with respect to the available average AES power for binary input loads
with different number of users $N=\{1,2,3\}$. We can observe that with more
than one user, we have different regimes of operation corresponding to the
number of users that receive energy from the grid. Similarly, in Fig.
\ref{fig:exp} we consider the scenario with exponential input loads. In both
models, regardless of the number of users in the system the total average
power consumed by the users is fixed to $P_{X}$. In the figures we set
$P_{X}=1$. As expected, if the average power provided by the AES is equal to
the total average power demanded by the users, perfect privacy can be achieved. Instead, as the average power of the
AES goes to zero, all the information is revealed to the UP, and thus, the
information leakage rate is equal to the sum of the entropies of all the input
loads. In between these two extremes the privacy-power function exhibits a
monotone decreasing convex behaviour, and the information leakage rate
increases with the number of users in the system.

\begin{figure}[t]
\psfrag{I(P) bits}{\footnotesize$\mathcal{I}(P)$ bits}
\psfrag{P}{\footnotesize$P$}
\centering\includegraphics[width=3in]{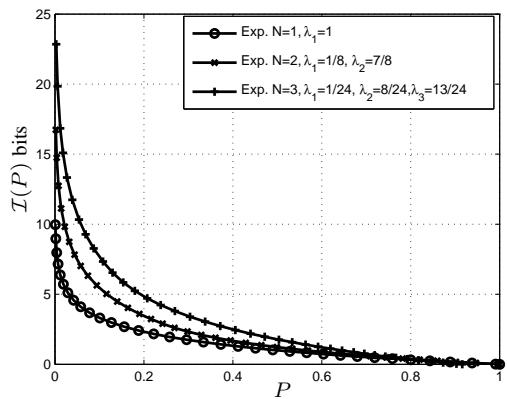}\caption{$\mathcal{I}%
(P)$ with respect to the average AES power $P$ for exponential input loads
with different number of users.}%
\label{fig:exp}%
\end{figure}

\section{Conclusions}

\label{sec:con}

We have introduced and studied the privacy-power function, $\mathcal{I}(P)$,
which characterizes the achievable information theoretic privacy in a
multi-user SM system in the presence of an AES. We have provided a
single-letter information theoretic characterization for $\mathcal{I}(P)$, and
showed that it can be evaluated numerically when the input loads are discrete.
We have also provided explicit characterization of the privacy-power function
for binary and exponential input load distributions. We have shown
that the optimal allocation of the energy provided by the AES in the
exponentially distributed input load scenario can be derived using the reverse
waterfilling algorithm, which resembles the rate-distortion function for
multiple Gaussian sources.

We believe that the proposed information theoretic framework for privacy in SM systems provides valuable tools to identify the fundamental challenges and limits for this critical problem, whose importance will only increase as SM adoption becomes more widespread. Many interesting research problems implore further studies, including time correlated input loads, systems with multiple EMUs, as well as cost and pricing issues considering dynamic pricing over time.


\bibliographystyle{ieeetran}
\bibliography{ref2}

\end{document}